\documentclass[12pt,oneside,reqno]{amsart}
\usepackage{extsizes}
\usepackage{apacite}
\usepackage{pgf}
\usepackage{amsaddr}
\usepackage{bbm}
\usepackage{url}
\usepackage{nicefrac}

\usepackage[colorlinks=true,citecolor=blue]{hyperref}
\usepackage[marginratio=1:1,height=584pt,width=450pt,tmargin=90pt]{geometry}
\usepackage{empheq}

\usepackage{amssymb}
\usepackage{amsmath}
\usepackage{mathtools}
\usepackage{subcaption}
\usepackage{graphicx}
\usepackage{wrapfig}
\usepackage{lipsum}
\usepackage{float}

\usepackage{upgreek}
\usepackage{bm}
\usepackage{cancel}

\usepackage{mdframed}
\usepackage{xfrac}
\usepackage{color,soul}
\usepackage{bbm}

\usepackage[capitalise,nameinlink,noabbrev]{cleveref}
\crefname{equation}{}{}
\crefname{enumi}{}{}
\Crefname{algocf}{Algorithm}{Algorithms}

\newtheorem{theorem}{Theorem}[section]

\newtheorem{lemma}[theorem]{Lemma}
\newtheorem{proposition}[theorem]{Proposition}
\newtheorem{definition}[theorem]{Definition}
\newtheorem{remark}[theorem]{Remark}

\usepackage[ruled,linesnumbered]{algorithm2e}

\numberwithin{equation}{section}

\usepackage{xspace}

\usepackage{multirow}
\usepackage{microtype}

\def\diffd{\mathrm{d}} 

\DeclareDocumentCommand\differential{ o g d() }{ 
	\IfNoValueTF{#2}{
		\IfNoValueTF{#3}
			{\diffd\IfNoValueTF{#1}{}{^{#1}}}
			{\mathinner{\diffd\IfNoValueTF{#1}{}{^{#1}}\argopen(#3\argclose)}}
		}
		{\mathinner{\diffd\IfNoValueTF{#1}{}{^{#1}}#2} \IfNoValueTF{#3}{}{(#3)}}
	}
\DeclareDocumentCommand\dd{}{\differential} 

\DeclareDocumentCommand\partialderivative{ s o m g g d() }
{ 
	\IfBooleanTF{#1}
	{\let\fractype\flatfrac}
	{\let\fractype\frac}
	\IfNoValueTF{#4}
	{
		\IfNoValueTF{#6}
		{\fractype{\partial \IfNoValueTF{#2}{}{^{#2}}}{\partial #3\IfNoValueTF{#2}{}{^{#2}}}}
		{\fractype{\partial \IfNoValueTF{#2}{}{^{#2}}}{\partial #3\IfNoValueTF{#2}{}{^{#2}}} \argopen(#6\argclose)}
	}
	{
		\IfNoValueTF{#5}
		{\fractype{\partial \IfNoValueTF{#2}{}{^{#2}} #3}{\partial #4\IfNoValueTF{#2}{}{^{#2}}}}
		{\fractype{\partial^2 #3}{\partial #4 \partial #5}}
	}
}

\DeclareDocumentCommand\pdv{}{\partialderivative} 

\newcommand{\p}{\mathbb{P}}
\newcommand{\reals}{\mathbb{R}}
\newcommand{\intr}{\int_{\reals}}
\newcommand{\ito}{It{\^o}}
\newcommand{\e}{\mathbb{E}}
\newcommand{\ee}{\mathrm{e}}
\newcommand{\iid}{i.i.d.\@\xspace}
\newcommand{\f}{\mathcal{F}}
\newcommand{\ds}{\dd{s}}
\newcommand{\dt}{\dd{t}}
\newcommand{\dx}{\dd{x}}
\newcommand{\dy}{\dd{y}}
\newcommand{\Lsc}{\mathcal{L}^{\mathrm{sr}}}

\newcommand{\eqdist}{\overset{\mathcal{D}}{=}}

\SetKw{KwAnd}{and}
\SetKwComment{Comment}{$\triangleright$\ }{}

\usepackage{tikz}
\usepackage{natbib}

\author{Young Lee\qquad\quad Patrick J. Laub\qquad\quad Thomas Taimre\\
\emph{\small{\hspace{-5mm}
H\MakeLowercase{arvard}}\hspace{26mm}
M\MakeLowercase{elbourne}\hspace{28mm} Q\MakeLowercase{ueensland}}\\
\vspace{8mm}
Hongbiao Zhao\qquad\qquad Jiancang Zhuang\\
\emph{\small{S\MakeLowercase{hanghai}}\quad\qquad\qquad\qquad\qquad\qquad T\MakeLowercase{okyo}}\qquad\qquad}

\date{\today\text{. To appear in Journal of Applied Probability, \textbf{59}(1), 2022}}

\begin{document}

\title{Exact simulation of extrinsic stress-release processes}

\maketitle




\begin{abstract}
	We present a new and straightforward algorithm that simulates exact sample paths for a generalized stress-release process. The computation of the exact law of the joint interarrival times is detailed and used to derive this algorithm. Furthermore, the martingale generator of the process is derived and induces theoretical moments which generalize some results of \citet{borovkov_vere-jones_2000} and are used to demonstrate the validity of our simulation algorithm.
\end{abstract}

\vspace{1mm}



\section{Introduction}

Stress-release processes are a class of point processes, the first of which were self-correcting processes~\citep{isham1979self}. Intuitively, a process is self-correcting if the occurrence of past points inhibits the occurrence of future points. The stress-release processes are a generalization of self-correcting processes and was introduced in a series of papers by Vere-Jones and others \citep{ZhengVereJones:1991,ZhengVereJones:1994} as well as extensions to coupled stress-release processes \citep{Liu:1998,Shi:1998,LuHarteBebbingon:1999} and further developments \citep{BebbingtonHarte:2001,BebbingtonHarte:2003}.

In this study, we work with a \emph{generalization} of the stress-release process which includes an exogenous point process term whose values upon arrivals are modeled by a \emph{positive} real-valued random variable. We call our model the \emph{extrinsic stress-release processes}.

We present a new formula for the law of the joint interarrival times for extrinsic stress-release processes. As a natural consequence, an exact simulation algorithm is then proposed which gives an alternative method to generating sample paths relative to standard methods~\citep{LewisShedler1979}. Our exact simulation algorithm naturally extends the results of~\citet{wang:1991} as a special case. The extension of our model is motivated by the influence of exogenous geophysical data on earthquake occurrence~\citep[see e.g.,][]{ZhuangMa:1998,OgataReview:2017}. Point process models of this kind are typically used to describe the evolution of stochastic phenomena in earthquake modeling and it is important to be able to simulate them for reliable predictions of damage due to a range of earthquake scenarios. Thinning algorithms~\citep{Ogata1981} have been successfully employed to simulate a wide range of point processes, such as inhomogeneous point processes~\citep[pp.\ 270--271 of][]{DaleyBook,zammit-pnas}, or Hawkes processes~\citep{veen-schoenberg}. Indeed the same idea can be applied to the generalized stress-release process proposed here. In this paper, simulation of the extrinsic stress-release process by our exact algorithm will be compared with the standard thinning algorithm.

Finally, we present the infinitesimal generator for extrinsic stress-release processes. This generator is intimately linked to the martingale problem which is used to characterize the weak solutions of partial integro-differential equations~\citep{LiptserShiryaev}, and it allows us to derive the theoretical reciprocal moments of the intensity function. In \cref{scn:numerical}, these reciprocal moments are used to demonstrate the correctness of our simulation algorithms. Basic notions and results in stochastic calculus are taken as prerequisites throughout the present text~\cite[confer e.g.,][]{Protter2005}.

\section{Extrinsic Stress-Release Model}%
\label{sec:gsrm}

At the base of everything is some filtered probability space $(\Omega,\f,\mathbb{F},\p)$. We assume that $\f_0$ is trivial, and the filtration $\mathbb{F} \coloneqq (\f_t)_{t\geq 0}$ fulfills the usual conditions and is generated by a point process $N(\cdot)$ on $\mathbb{R}_{+}$ where $0 < T_1 < T_2 < \cdots$ denote the occurrence times of the events. Let $N_t = \sharp\{T_i : 0 < T_i \le t\}$ be the number of the occurrence points in the time interval $(0,t]$ with $N_0=0$. Furthermore, we let $N'_t = \sharp\{T'_j : 0 < T'_j \le t\}$ be a Poisson process on $\mathbb{R}_{+}$ with arrival times $0 <T^{\prime}_1 < T^{\prime}_2< \cdots$  endowed with intensity $\rho$ and is independent of $N_t$, with $N^{\prime}_0=0$.

\begin{definition}%
	\label{def:extrensic-stress-release}
	The proposed extrinsic stress-release process $N(\cdot)$ is a point process on $\reals_+$ with conditional intensity function given by
	\begin{align}
		\label{eq:conditional intensity function original}
		\lambda_t \coloneqq \lambda(t\,|\,\f_t) = \lambda_0\exp(\beta t - S_t - S'_t),\,\quad t\geq 0,
	\end{align}
	where $S_t=\sum_{i\,:\,T_i< t} X_i$ and $S'_t=\sum_{j\,:\,T'_j < t} Y_j$ are the compound point process and compound Poisson process, respectively. The $X_i$'s and $Y_j$'s are \iid positive random variables, with distribution functions $F_X$ and $F_Y$ respectively, and the stress accumulation rate is the constant $\beta>0$. $\hfill \diamond$
\end{definition}

\begin{figure}[!t]
	\centering
	\input{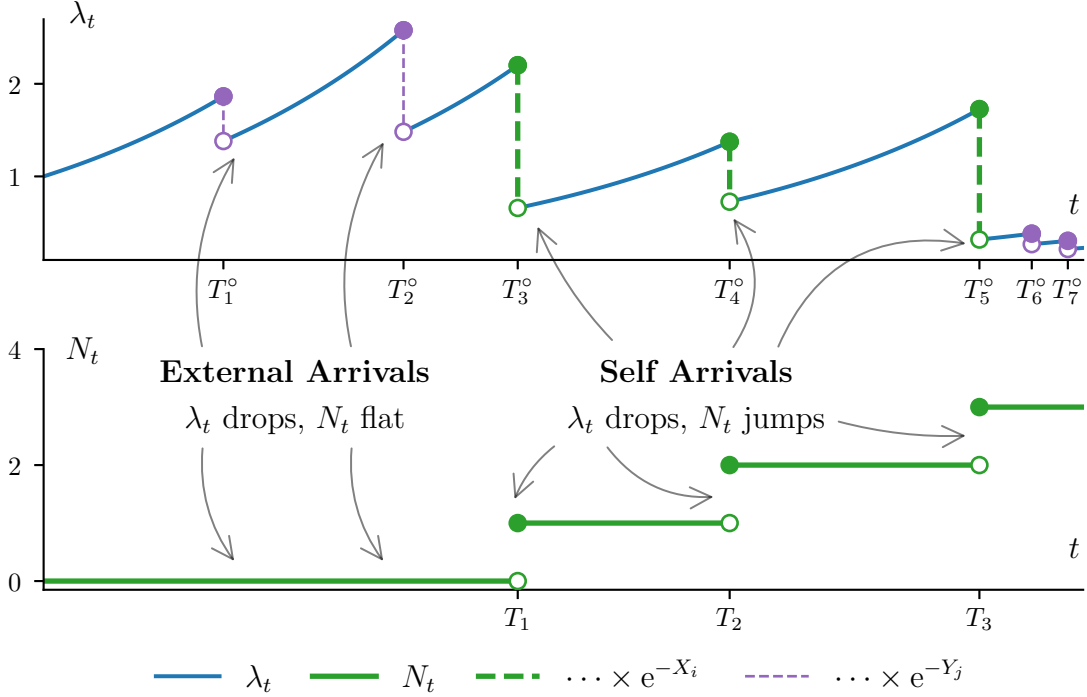}
	\caption{An example realization of an extrinsic stress-release process, with $\lambda_0 = 1$, $\beta = 1.5$, $\rho = 2$, and $X_i \sim \mathsf{Exp}(1)$ and $Y_j \sim \mathsf{Exp}(2)$. Note that $N$ is c\`adl\`ag while $\lambda$ is c\`agl\`ad.}%
	\label{fig:example}
\end{figure}

Between jumps, $\lambda_t$ increases exponentially with a positive rate of $\beta>0$. Jumps are downward multiplicative factors of size $\ee^{-X_i} < 1$ for a self-arrival at time $T_i$, or of size $\ee^{-Y_j} < 1$ for an external arrival at time $T'_j$. When a self-arrival occurs at time $T_i$, $N_t$ increases by one, hence $(N_t, \lambda_t)$ is a Markov process. Instead of separating the self-arrivals of $N_t$ and the external arrivals of $N'_t$, it is sometimes convenient to consider all the arrivals indiscriminately. As such, we label the $k$-th arrival as $T^\circ_k$, and it can correspond either to some self-arrival $T_i$ or some external arrival $T'_j$. See \cref{fig:example} for an example realization of the extrinsic stress-release process, with the effects of the self-arrivals and external arrivals on the conditional intensity function $\lambda_t$ highlighted.

\begin{remark}
	Our proposed extrinsic stress-release process differs slightly from the coupled stress-release model of~\citet{Liu:1998}. Their equivalent of the $X_i$ and $Y_j$ variables in \cref{def:extrensic-stress-release} are not unobserved random variables, they are deterministic functions of the observed earthquake magnitudes. In our proposed model, we allow for these quantities to take any \iid random variables which are positive and unobserved, so our formulation generalizes theirs. Differently from ours, they allow for model parameters $c$ and $c^{\prime}$ in the exponent of equation \cref{eq:conditional intensity function original} of the form
	\begin{align}
		-c S_t - c^{\prime} S^{\prime}_t
	\end{align}
	where $c$ and $c^{\prime}$ can \emph{either} take negative or positive values, thereby allowing for both damping and excitation. We only consider the inhibitory regime, i.e.\ $c=c^{\prime}=1$, hence in this case, their formulation for general $c$ and $c^{\prime}$ subsumes ours. In either formulation, little or no work has appeared on exact simulation strategies for the coupled stress-release model. We further add some new aspects to the computation of explicit generators which facilitates moment computations. For other theoretical and stationary moment calculations without the exogenous term $S'$, see~\citet{vere1984moments,ogata1984inference,verejones-variance:1988,borovkov_vere-jones_2000}.
	$\hfill \diamond$
\end{remark}

\section{The law of joint interarrival times}

In this section, we present the explicit law of the joint interarrival times for extrinsic stress-release processes. This terminology, `joint interarrival time', refers to the time between each of $T^\circ_k$ arrivals (defined in \cref{sec:gsrm}).
\ito's formula \citep{JacodShiryaev} splits $\lambda_t$ into continuous and jump components,
\[
	\lambda_t = \lambda_0 + \int_0^t\beta\lambda_s \ds + \sum_{i : T_i \le t} \lambda_{T_i} (\ee^{-X_i}-1) + \sum_{j : T'_j \le t} \lambda_{T'_j} (\ee^{-Y_j}-1) .
\]
Between consecutive jumps the $\lambda$ process evolves according its continuous part. In particular, conditioned on $T^\circ_k$ and $\lambda_{T^{\circ+}_k}$, we have
\begin{equation} \label{eq:continuous_part}
	\lambda_t = \lambda_{T^{\circ+}_k} \exp\bigl(\beta(t-T^\circ_k)\bigr) \quad \text{for } t \in (T^\circ_k, T^\circ_{k+1}).
\end{equation}
The intensity of the $T^\circ_k$ arrivals is the combination of the $T_i$ and $T'_j$ arrival intensities $\lambda_t + \rho$.
Let the $k$-th joint interarrival time be denoted by $\tau_k \coloneqq T^\circ_k-T^\circ_{k-1}$ with cumulative density function $F_{\tau_{k}}$. With \cref{eq:continuous_part}, we can simplify the point process relation
\begin{align} \label{eq:joint-iit-law}
	F_{\tau_{k+1}}(t)
	= 1 - \exp\Bigl( {-} \int_{0^+}^t (\lambda_{T^\circ_k + s} + \rho) \ds \Bigr)
	= 1 - \exp\Bigl( {-} \frac{\lambda_{T^{\circ+}_k}}{\beta} (\ee^{\beta t}-1) \Bigr) \ee^{-\rho t} ,
\end{align}
which is the law of the joint interarrival times.

\section{Simulation methods}%
\label{scn:simulation}

The law of the joint interarrival times in \cref{eq:joint-iit-law} can be used to derive a simulation method for extrinsic stress-release processes. To simulate we need to: (i) generate $T^\circ_k$ joint interarrival times, and (ii) be able to attribute each arrival as being either a self-arrival from $N_t$ or an external arrival from $N'_t$. By the inverse probability integral transform, we have $\tau_{k+1} \eqdist F_{\tau_{k+1}}^{-1}(U),\,U\sim \mathsf{U}[0,1]$ where~$\eqdist$ denotes equality in distribution. The inverse $F_{\tau_{k+1}}^{-1}$ does have an analytic solution (which is somewhat rare) in terms of the Lambert~$W$ function, so we can generate joint interarrival times by the inverse transform method. However the Lambert~$W$ function is relatively slow in many software packages, and this calculation does not perform the second attribution step. A faster alternative, which solves both problems at once, is to use the \emph{composition method}.

\subsection{Exact simulation of stress-release model.}

The composition method~\citep[Section~VI.2.3]{devroye1986} simulates $\tau_{k+1}$ from two simpler independent random variables $\tau^{(1)}_{k+1}$ and $\tau^{(2)}_{k+1}$ by taking
\begin{align}
	\tau_{k+1} \eqdist \tau^{(1)}_{k+1} \wedge \tau^{(2)}_{k+1}
\end{align}
where the notation $\tau^{(1)}_{k+1} \wedge \tau^{(2)}_{k+1}$ is simply shorthand for $\min\{ \tau^{(1)}_{k+1}, \tau^{(2)}_{k+1} \}$.
One way to satisfy this relation is for $\p(\tau^{(1)}_{k+1}>s) = \exp\bigl({-} \lambda_{T^{\circ+}_k} \beta^{-1} (\ee^{\beta s}-1)\bigr)$ and $\p(\tau^{(2)}_{k+1}>s)=\ee^{-\rho s}$, so
\begin{align}
	\label{eq:laws-self-external}
	\tau^{(1)}_{k+1} & \eqdist \frac{1}{\beta}\log\Bigl(1-\frac{\beta}{\lambda_{T^\circ_k}}\log(U_1)\Bigr),\quad\tau^{(2)}_{k+1} \eqdist -\frac{1}{\rho}\log(U_2),\quad U_1,U_2\sim \textsf{U}[0,1].
\end{align}
This is the key step in the composition algorithm, presented in full in \cref{alg:composition}.

\begin{algorithm}[ht]
	\caption{Generate an extrinsic stress-release process by composition.}%
	\label{alg:composition}
	\small
	\KwIn{start intensity $\lambda_0$, stress rate $\beta$, external arrival rate $\rho$, jump size distributions $F_X$ and $F_Y$, end time $T$}
	\Begin{
	Initialize $T^\circ_0 \gets 0$, $\epsilon \gets 10^{-10}$ or similar, $\lambda_{\epsilon} \gets \lambda_0$, $i \gets 0$, $j \gets 0$, $k \gets 0$ \;
	Simulate $\tau^{(1)}_{k+1}$ and $\tau^{(2)}_{k+1}$ via equation \cref{eq:laws-self-external}
	and let $T^\circ_{k+1} \gets T^\circ_k + \tau^{(1)}_{k+1} \wedge \tau^{(2)}_{k+1}$ \;\label{step:composition-start}
	\If{$T^\circ_{k+1} > T$}{
		\KwRet self-arrivals $\{ T_1, \dots, T_i \}$ and, if desired, external arrivals $\{ T'_1, \dots, T'_j \}$
	}
	Calculate $\lambda_{T^\circ_{k+1}} \gets \lambda_{T^\circ_{k} + \epsilon} \exp\bigl( \beta (T^\circ_{k+1} - T^\circ_k) \bigr)$ by \cref{eq:continuous_part} \;
	\uIf{$\tau^{(1)}_{k+1} < \tau^{(2)}_{k+1}$}{
		Update $i \gets i + 1$ and $T_i \gets t + \tau^{(1)}_{k+1}$ \Comment*[r]{Self-arrival}
		Simulate $X_i \sim F_X$ and update $\lambda_{T^\circ_{k+1} + \epsilon} \gets \lambda_{T^\circ_{k+1}} \ee^{-X_i}$ \;
	}\Else{
		Update $j \gets j + 1$ and $T'_j \gets t + \tau^{(2)}_{k+1}$ \Comment*[r]{External arrival}
		Simulate $Y_j \sim F_Y$ and update $\lambda_{T^\circ_{k+1} + \epsilon} \gets \lambda_{T^\circ_{k+1}} \ee^{-Y_j}$ \;
	}
	Update $k \gets k + 1$ and go to line~\ref{step:composition-start} \;
	}
\end{algorithm}

\subsection{Simulation by thinning}

Extrinsic stress-release processes can also be simulated via the \emph{thinning algorithm}. The basic idea in this method is to generate a point process which has more arrivals than the model dictates, then probabilistically remove the excess points. The result can be computationally inefficient, and we compare the runtime of the thinning and composition simulation methods in \cref{scn:numerical}.

The first step in the thinning algorithm is to generate the $N'_t$ and $S'_t$ processes. The self-arrivals are then generated conditional on these external arrivals. Each self-arrival is generated sequentially, and requires a local upper bound on the intensity process. If we know $S'_t$ for all $t \in \reals_+$ we obviously have
\[ \lambda_t \leq \lambda_0\exp(\beta t - S'_t),\,\quad t\in\reals_+, \]
though as $t$ increases this becomes an extremely loose bound. However, if we also know the process $S_t$ up until time $\tau$, then
\begin{align}
	\label{eq:intensity-bounds}
	\overline{\lambda}_{t \mid \tau} \coloneqq \lambda_0\exp(\beta t - S_{t \wedge \tau} - S'_t),\,\quad t\in\reals_+,
\end{align}
is a much tighter upper bound on the intensity, at least for $t \in (\tau, \tau+\Delta)$ for moderately small $\Delta$'s. \Cref{fig:intensity-bounds} shows some example realizations of~\cref{eq:intensity-bounds}.

With this definition, we can describe the thinning algorithm for the generalized stress-release process in \cref{alg:thinning}. In particular, line~\ref{step:thin-start} of the algorithm uses \cref{eq:intensity-bounds} to find an upper bound of $\lambda_t$ over a small region $t \in (\tau, \tau+\Delta]$; these maximum values are not too tedious to find, as they occur either at the end time $\tau+\Delta$ or at one of the external arrival times $T'_j$ which arrives inside the region.

\begin{figure}[!t]
	\centering
	\input{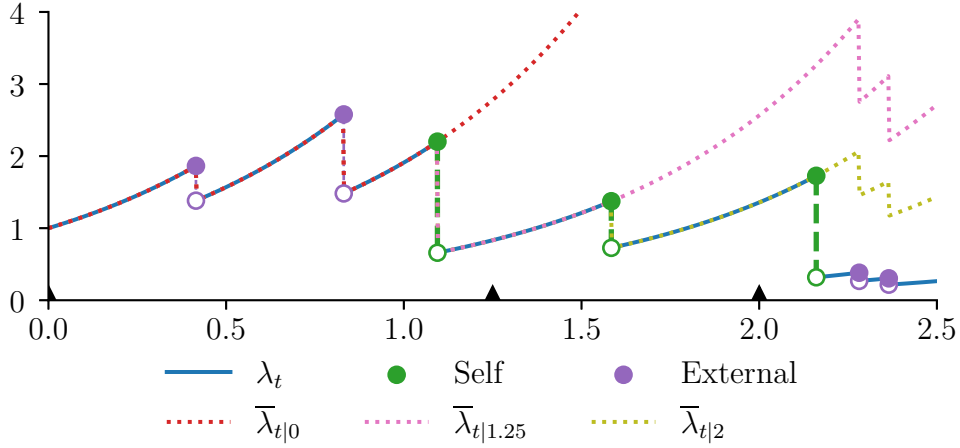} \\
	\vspace{-1em}
	\caption{Example $\overline{\lambda}_{t \mid \tau}$ upper bounds on the intensity function $\lambda_t$. This is the same realization of the generalized stress-release process from \cref{fig:example}.}%
	\label{fig:intensity-bounds}
\end{figure}

\begin{algorithm}[!h]
	\caption{Generate an extrinsic stress-release process by thinning.}%
	\label{alg:thinning}
	\small
	\KwIn{start intensity $\lambda_0$, stress rate $\beta$, external arrival rate $\rho$, jump size distributions $F_X$ and $F_Y$, end time $T$, step size $\Delta$}
	\Begin{
		Simulate the external arrivals $\{T'_1, \dots, T'_{N'_T}\}$ at rate $\rho$ by standard Poisson process methods \;
		Simulate \iid external jump sizes $\{Y_1, \dots, Y_{N'_T}\}$ from $F_Y$ and construct the $(S'_t\,)_{t\in[0,T]}$ process \;
		Initialize $i = 0$, $t = 0$, $\epsilon \gets 10^{-10}$ or similar \;
		Set $M$ to be the maximum value of $\overline{\lambda}_{s \mid t + \epsilon}$, cf.\ \cref{eq:intensity-bounds}, over $s \in (t, t + \Delta]$ \; \label{step:thin-start}
		Generate a proposal self-arrival $T^* = t + E$ where $E \sim \mathsf{Exp}(M)$ \;
		\uIf{$T^* > T$}{
			\KwRet self-arrivals $\{ T_1, \dots, T_i \}$
		} \ElseIf{$T^* > t + \Delta$}{
			Reject the proposal, set $t \gets t + \Delta$, go to line~\ref{step:thin-start} \;
		}
		Sample $U \sim \mathsf{U}[0,1]$ \;
		\uIf{$U \le \overline{\lambda}_{T^* \mid t} \,/\, M$}{
			Accept the proposal: $i \gets i + 1$, $T_i \gets T^*$ \;
			Simulate a jump size $X_i \sim F_X$ and update $\lambda_{T_i + \epsilon} \gets \lambda_{T_i} \ee^{-X_i}$ \;
			Set $t \gets T^*$, go to line~\ref{step:thin-start} \;
		}
		\Else{
			Reject the proposal: $t \gets t + \Delta$, go to line~\ref{step:thin-start} \;
		}
	}
\end{algorithm}

\section{The generator}%
\label{scn:generator}

In this section, we derive the explicit form of the infinitesimal generator for our process. With this, we are able to find reciprocal moments which is then used to confirm the validity of our simulation algorithm.

\subsection{Constructing the infinitesmal generator}

Let us introduce the integro-differential operator $\Lsc$ of our extrinsic stress-release process $(\lambda_t,N_t,t)$ which acts on a function $f(\lambda,n,t)$ within its domain $\Omega(\Lsc)$ as follows
\begin{equation}
	\label{eq:operator-sc-definition}
	\begin{aligned}
		\Lsc f \coloneqq \pdv{f}{t} + \beta \lambda \pdv{f}{\lambda} + \lambda & \intr [f(\lambda \ee^{-x},n+1,t)-f(\lambda,n,t)] \, F_X(\dx) \\
		+\rho                                                                  & \intr [f(\lambda \ee^{-y},n,t)-f(\lambda,n,t)] \, F_Y(\dy).
	\end{aligned}
\end{equation}
From Propositions II.1.16 \& II.1.15 in~\citet{JacodShiryaev}, the conditional intensity function in equation~\cref{eq:conditional intensity function original} can be recast as
\begin{align}
	\lambda_t = \lambda_0\exp\left(\beta t - \int_0^t\intr x \mu(\dx, \ds) - \int_0^t\intr y \mu'(\dy, \ds)\right)
\end{align}
where $\mu$ and $\mu'$ are the jump measures associated with $S$ and $S'$, respectively. Their associated predictable compensators are $\nu(\dx, \dt) = F_X(\dx) \lambda_t \dt$ and $\nu'(\dy, \dt) = F_Y(\dy) \rho \dt$. We now state the following result:
\begin{proposition}%
	\label{prop:generator}
	Let the integro-differential operator for our stress-release process be defined as in equation~\cref{eq:operator-sc-definition}. Then for each $t\in\reals_+$ the following holds
	\begin{align}
		\label{eq:ic3}
		\e[f(\lambda_t,N_t,t)] = f(0,N_0,\lambda_0) + \e\left[\int_0^t \Lsc f(\lambda_s,N_s,s) \ds\right]
	\end{align}
	if the following $f$-integrability conditions hold:
	\begin{align}
		\label{eq:ic1}
		\e\left[\int_0^t \lambda_s \ds \intr [f(\lambda_{s^+},N_s,s)-f(\lambda_s,N_{s^-},s)]^2 \, F_X(\dx) \right]<\infty
	\end{align}
	and
	\begin{align}
		\label{eq:ic2}
		\e\left[\int_0^t \rho \ds \intr [f(\lambda_{s^+},N_s,s)-f(\lambda_s,N_{s^-},s)]^2 \, F_Y(\dy) \right]<\infty .
	\end{align}
	Moreover, $f$ satisfies the following partial integro-differential equation
	\begin{equation}
		\begin{aligned}
			\label{eq:pide-gsrm}
			 & \pdv{f}{t} + \beta \lambda \pdv{f}{\lambda} + \lambda \intr\left[f(\lambda \ee^{- x},n+1,t)-f(\lambda,n,t)\right] F_X(\dx) \\
			 & \qquad\qquad +\rho\intr\left[f(\lambda \ee^{-y},n,t)-f(\lambda,n,t)\right] F_Y(\dy)=0.
		\end{aligned}
	\end{equation}
\end{proposition}

\begin{proof}
	First note that $\lambda_t$ can be recast as
	\begin{align}
		\lambda_{t} = \lambda_0 + \int_0^t \beta\lambda_{s} \ds + \int_0^t\intr\lambda_s(\ee^{- x}-1)\mu(\dx, \ds)+ \int_0^t\intr\lambda_s(\ee^{- y}-1)\mu'(\dy, \ds).
	\end{align}
	Invoking \ito's formula~\citep[confer][Chapter VI]{medvegyev:2009} on the arbitrary function $f(\lambda_t,N_t,t)$ yields
	\begin{align*}
		f(\lambda_t,N_t,t) & = f(\lambda_0,N_0,0) + \int_0^t \pdv{f}{s} \ds + \int_0^t \pdv{f}{\lambda} \dd{\lambda_s^{c}} + \sum_{0<s\leq t} [f(\lambda_{s^+},N_s,s) - f(\lambda_s,N_{s^-},s)]
	\end{align*}
	where $\lambda^c$ denotes the continuous part of the semimartingale. We can write
	\begin{align*}
		&\sum_{0<s\leq t}[f(\lambda_{s^+},N_s,s) - f(\lambda_s,N_{s^-},s)]
		 \\
		 & = \sum_{0<s\leq t}[f(\lambda_{s^+},N_s,s) - f(\lambda_s,N_{s^-},s)]\cdot\Delta N_s          \\
		 & \qquad+ \sum_{0<s\leq t}[f(\lambda_{s^+},N_s,s) - f(\lambda_s,N_{s^-},s)]\cdot\Delta N'_s         \\
		 & = Q_t(f) + \int_0^t\intr [f(\lambda_{s^+},N_s,s) - f(\lambda_s,N_{s^-},s)]\,\nu(\dx, \ds)   \\
		 & \qquad+ Q'_t(f) + \int_0^t\intr [f(\lambda_{s^+},N_s,s) - f(\lambda_s,N_{s^-},s)]\,\nu'(\dx, \ds)
	\end{align*}
	with $Q_{\cdot}(f)$ and $Q'_{\cdot}(f)$ being defined by $Q_{\cdot}(f) \coloneqq \int_0^{\cdot} \intr [f(\lambda_{s^+},N_s,s) - f(\lambda_s,N_{s^-},s)](\mu-\nu)(\dx, \ds)$ and $Q'_{\cdot}(f) \coloneqq \int_0^{\cdot} \intr [f(\lambda_{s^+},N_s,s) - f(\lambda_s,N_{s^-},s)](\mu'-\nu')(\dx, \ds)$, respectively. The integrability conditions in equations \cref{eq:ic1,eq:ic2} guarantee that $Q(f)$ and $Q'(f)$ are square integrable martingales \citep[Theorem VIII of Chapter II]{Bremaud1981}. Hence we conclude that the process $(f(\lambda_t,N_t,t))_{t\in\reals_+}$ is a special semimartingale \citep{Protter2005} which can be decomposed into a martingale and a predictable finite variation process
	\begin{align}
		\label{eq:tmp2}
		f(\lambda_t,N_t,t)-f(\lambda_0,N_0,0) = Q_t(f) + Q'_t(f)+\int_0^t \Lsc f(\lambda_s,N_s,s) \,{\ds}.
	\end{align}
	Since $Q(f)$ and $Q'(f)$ are square-integrable martingales, the process defined by
	\begin{align}
		\label{eq:dynkin}
		\left(f(\lambda_t,N_t,t)-f(\lambda_0,N_0,0)-\int_0^t \Lsc f(\lambda_s,N_s,s) \,{\ds} \right)_{t\in\reals_+}
	\end{align}
	is also a square-integrable martingale. Taking the expected value of both sides on equation \cref{eq:tmp2} yields the result.
	For the subsequent expression, first define $T > t$ and $g_t \coloneqq \e[h(\lambda_T,N_T,T)\,|\,t,\lambda_t=\lambda,N_t=n]$ for some function $h$ such that $g$ satisfies the $g$-integrability conditions \cref{eq:ic1,eq:ic2}. Then by construction, $g_t$ is a martingale. By similar arguments, we see that $g-Q(g)-Q'({g})$ is a square-integrable martingale, but $g - Q(g)-Q'(g) = \int_0^\cdot \Lsc g \ds$ is also a continuous process with finite variation. It must therefore be a continuous martingale with finite variation \citep[Corollary I-3.16]{JacodShiryaev}, hence, we must have $\Lsc g=0$ $\p-$almost surely which yields the partial integro-differential equation in \cref{eq:pide-gsrm}.
\end{proof}

\subsection{Reciprocal moments}%
\label{ssn:reciprocal}

Consistent with the observation in~\citet{vere1984moments}, we are unable to find moments of $\lambda_t$ but we can find moments of its reciprocal. Throughout \cref{ssn:reciprocal,ssn:covariance} we assume that $\lambda_0 = 1$, though the same arguments can be made in the general $\lambda_0$ case.

\begin{lemma}
	\label{eq:E1/L}
	Let $m_1^S \coloneqq \int (\ee^{x}-1) \, F_X(\dx)$ and $m_1^E \coloneqq \int (\ee^{y}-1) \, F_Y(\dy)$. We assume that $\beta>\rho m_1^E$, and $\lambda_0 = 1$. Then the expectation of $\lambda^{-1}_t$ is given by
	\begin{align}
		\label{inv1}
		\e[{\lambda^{-1}_t}] = \ee^{-\psi_1 t} + \frac{m_1^S}{\psi_1}(1-\ee^{-\psi_1 t})
	\end{align}
	where $\psi_1 \coloneqq \beta-\rho m_1^E$.
\end{lemma}
\begin{proof}
	From \cref{prop:generator}, we have for $f\in\Omega(\Lsc)$ that
	\[ f(\lambda_t,N_t,t) - f(\lambda_0,N_0,0) - \int \Lsc f(\lambda_s,N_s,s) \ds \]
	is an $\f-$martingale. Setting $f=\lambda^{-1}$ in the generator yields
	\begin{align}
		\Lsc (\lambda^{-1}) = -{\beta}{\lambda^{-1}} + m_1^S + {\rho}{\lambda^{-1}} m_1^E,
	\end{align}
	and $\e\bigl[\lambda^{-1}_t-\lambda^{-1}_0-\int_0^t \Lsc (\lambda^{-1}_s) \ds \bigr]=0$. Differentiating $\theta_1(t) \coloneqq \e[{\lambda^{-1}_t}]$ with respect to $t$ yields the non-linear inhomogeneous ODE
	\begin{align}
		\theta'_1 + \psi_1 \theta_1=m_1^S,\quad\theta_1(0)=1
	\end{align}
	whose solution is given in equation \cref{inv1}.
\end{proof}

By a similar token, and setting $f=\lambda^{-2}$, we state the following:
\begin{lemma}
	\label{eq:E1/L2}
	Let $m_2^S \coloneqq \int (\ee^{2x}-1) \, F_X(\dx)$, $m_2^E \coloneqq \int (\ee^{2y}-1) \, F_Y(\dy)$. We assume that $\beta>\rho m_1^E$, $2 \beta > \rho m_2^E$, and $\lambda_0 = 1$. Then the expectation of $\lambda^{-2}_t$ is given by
	\begin{align}
		\label{inv2}
		\e\bigl[{\lambda^{-2}_t}\bigr]
		= \ee^{-\psi_2 t} + m^S_2 \Bigl\{
		\frac{\ee^{-\psi_1 t} - \ee^{-\psi_2 t}}{\psi_2 - \psi_1} +
		\frac{m_1^S}{\psi_1} \Bigl[\Bigl(\frac{1}{\psi_2} - \frac{\ee^{-\psi _1 t}}{\psi_2 - \psi_1}\Bigr) -
			\ee^{-\psi_2 t} \Bigl(\frac{1}{\psi_2} - \frac{1}{\psi_2 - \psi_1}\Bigr)\Bigr]\Bigr\}
	\end{align}
	where $\psi_2 \coloneqq 2\beta - \rho m_2^E$.
\end{lemma}
We end this section by giving some recursive relationships related to the inverse moments of our process. Let $\mathbb{N}$ be the set of natural numbers and let $k\in\mathbb{N}$, $m_k^S \coloneqq \int (\ee^{kx}-1) \, F_X(\dx)$, $m_k^E \coloneqq \int (\ee^{ky}-1) \, F_Y(\dy)$ and $\psi_k:=k\beta-\rho m^E_k$. We further assume that $k\beta>\rho m^E_k$.
Then the generator for the function $f=\lambda^{-k}$ is readily computed as follows
\begin{align*}
	\Lsc(\lambda^{-k})=m_k^S\lambda^{k-1} - \psi_k\lambda^{k}.
\end{align*}
By the martingale property we have that
\begin{align}
	\label{eq:tmpy}
	\e\Bigl[\lambda^{-k}_t-\lambda^{-k}_0-\int_0^t \Lsc (\lambda^{-k}_s) \ds \Bigr]=0.
\end{align}
Define the quantity $\theta_k(t) \coloneqq \e[{\lambda^{-k}_t}]$ and differentiating equation~\cref{eq:tmpy} with respect to $t$, we arrive at the recursive non-linear inhomogeneous ODE
\begin{align*}
	\theta^{\prime}_k(t)+\psi_k \theta_k(t) = m^S_k\theta_{k-1}(t)
\end{align*}
where $\theta_0\equiv 1$, which can be solved in a recursive fashion to obtain further reciprocal moments.

\subsection{Covariance process}
\label{ssn:covariance}

We provide an expression for the mean of the product of two reciprocals of intensities for our process. This can be used to compute the covariance process. For $s<t$, it holds true that
\begin{align}
	\e[\lambda_t^{-1}\lambda_s^{-1} ] & =\e[\e[\lambda_t^{-1}\lambda_s^{-1}\,|\,\lambda_s^{-1}] ]\nonumber                                                                     \\
	                                                        & =\ee^{-\psi_1 (t-s)}\e[\lambda_s^{-2}]+\frac{m^S_1}{\psi_1}(1-\ee^{-\psi_1(t-s)})\e[\lambda^{-1}_s].\label{eq:cov}
\end{align}
To see why this is true, note that the inner expectation can be computed as follows:
\begin{align*}
	\e[\lambda_t^{-1}\lambda_s^{-1}|\lambda_s^{-1}] & =\lambda_s^{-1}\e[\lambda_t^{-1}|\lambda_s^{-1}]=\lambda_s^{-2}\ee^{-\psi_1 (t-s)}+\lambda^{-1}_s\frac{m^S_1}{\psi_1}(1-\ee^{-\psi_1(t-s)}).
\end{align*}
Therefore, for $s < t$ we have
\begin{align*}
	\mathbb{C}\mathrm{ov}(\lambda_s^{-1}, \lambda_t^{-1})
	&= \e[\lambda_t^{-1}\lambda_s^{-1}  ] - \e[\lambda_t^{-1}  ] \e[\lambda_s^{-1}   ] \\
	&= \ee^{-\psi_1 (t-s)}\e[\lambda_s^{-2}]+\frac{m^S_1}{\psi_1}(1-\ee^{-\psi_1(t-s)})\e[\lambda^{-1}_s] - \e[\lambda_t^{-1}  ] \e[\lambda_s^{-1}  ]
\end{align*}
where the remaining expectations are given by \cref{inv1,inv2}.
Furthermore, when we set $S^{\prime}=0$ in equation~\cref{eq:conditional intensity function original}, we retrieve the covariance results in \cite[Theorem 2, p.\ 317]{borovkov_vere-jones_2000} upon substituting the results of~\cref{eq:E1/L,eq:E1/L2} to get an expression for equation~\cref{eq:cov} and subsequently subtracting the quantity $\e[\lambda^{-1}_t]\cdot\e[\lambda^{-1}_s]$.

\section{Numerical results}%
\label{scn:numerical}

To confirm that the simulation algorithms of~\cref{scn:simulation} agree with the reciprocal moments derived in~\cref{scn:generator}, we compare the first reciprocal moments given by the theory to Monte Carlo estimates. The results are given in~\cref{fig:inverse-moments}. The theory and the simulated values agree nicely. This example is also used to illustrate the speed benefits to the composition simulation method over the thinning method. In particular, \cref{tbl:runtimes} shows how long the two algorithms took to generate a fixed number of realizations of the extrinsic stress-release process. The difference in the performance can be explained by: (i) the fact that the composition method is easily vectorized while thinning is not, and (ii) the thinning algorithm is inherently inefficient as it intentionally generates too many points and discards a possibly large fraction of them.

\newpage

\begin{figure}[h!]
	\centering
	\begin{tikzpicture}
		\draw (0, 0) node {\input{inverse-moments.pgf}};
		\draw (-5.4, 3) node {\large $\e[{\lambda^{-1}_t}]$};
		\draw (-1.0, -1.5) node {\large $t$};
		\draw (2.2, 3) node {\large $\e[{\lambda^{-2}_t}]$};
		\draw (6.5, -1.5) node {\large $t$};
	\end{tikzpicture}
	\vspace{0.2em}
	\caption{First reciprocal moments $\e[{\lambda^{-1}_t}]$ and second reciprocal moments $\e[{\lambda^{-2}_t}]$ for $t \in [0, 50]$ of an extrinsic stress-release process with $\lambda_0 = 1$, $\beta = 0.25$, $\rho = 1.25$, and $X_i \sim \mathsf{Exp}(3)$ and $Y_j \sim \mathsf{Exp}(10)$. The theoretical values, given by \cref{inv1,inv2}, are compared with crude Monte Carlo estimates using the two simulation methods from \cref{scn:simulation}. Both simulation methods were allocated the same amount of computation time (55--60 seconds); as such, \cref{alg:composition} generated 375,000 sample paths of $\lambda_t$, whereas \cref{alg:thinning} generated 25,000. The shaded regions indicate the 95\% confidence intervals of the Monte Carlo estimates.}%
	\label{fig:inverse-moments}
\end{figure}

\begin{table}[!h]
	\centering
	\begin{tabular}{lcccc}
		\multirow{2}{*}{\textbf{Algorithm}} & \multicolumn{4}{c}{\textbf{Number of realizations}}                            \\
		                                    & $10^2$                                              & $10^3$ & $10^4$ & $10^5$ \\ \hline
		Composition                         & 0.0205                                              & 0.0544 & 0.3600 & 3.2992 \\
		Thinning                            & 0.3026                                              & 3.0329 & 30.761 & 310.32 \\
	\end{tabular}

	\vspace{0.5em}
	\caption{Comparison of runtime (in seconds) to simulate a number of realizations of the extrinsic stress-release process until time $T = 100$ using \cref{alg:composition,alg:thinning}. The fastest time of 3 attempts is recorded. A grid search is performed to select the optimal step size $\Delta = 1.86$ for \cref{alg:thinning} (this search is not included in the runtimes). The specification of the process ($\lambda_0$, $\beta$, $\rho$ etc.) is the same as in \cref{fig:inverse-moments}.}%
	\label{tbl:runtimes}
\end{table}

The Python code used to generate these numerical results is available on Github:\\
{\footnotesize \url{https://github.com/Pat-Laub/exact-simulation-of-extrinsic-stress-release-processes}.}

\section{Concluding remarks}

In this article we have introduced a straightforward but computationally efficient way of simulating exactly for a class of generalized stress-release process. The idea stems from the observation that between contiguous jumps at points in time, the process satisfies the continuous part of the semimartingale and thus is governed by an ordinary differential equation. This permits us to derive an expression for the distribution between events. The end result is that we are able to sidestep the need to resort to thinning algorithms for the simulation of this class of point processes.

The explicit form of the infinitesimal generator for the extrinsic stress-release process is given. Theoretical reciprocal moments are derived which are then used to establish the validity of our simulation algorithms.

We envisage that the approach outlined in this paper extends naturally to the general coupled stress release models, i.e.\ the linked stress-release model \citep{BebbingtonHarte:2001,BebbingtonHarte:2003} with appropriate structures satisfying the martingale generators. Ongoing work is investigating such a problem.





\bibliographystyle{apacite}

\end{document}